\documentclass[3p,preprint]{elsarticle}

\journal{Information and Computation}

\usepackage[figuresright]{rotating}
\newtheorem{theorem}{Theorem}
\newtheorem{lemma}[theorem]{Lemma}
\newtheorem{proposition}[theorem]{Proposition}
\newtheorem{corollary}[theorem]{Corollary}
\newtheorem{definition}[theorem]{Definition}

\newtheorem{example}{Example}
\newtheorem{remark}{Remark}

\newproof{proof}{Proof}

\usepackage[utf8]{inputenc}
\usepackage[T1]{fontenc}
\usepackage{amssymb}
\usepackage{amsmath}
\usepackage{amsfonts}
\usepackage{xifthen}
\usepackage{graphicx}
\usepackage{enumerate}
\usepackage{tikz}
\usetikzlibrary{petri}
\usepackage{multicol}

\usepackage{complexity}
\usepackage{mathtools}
\usepackage{comment}
\usepackage{mathrsfs}
\usepackage{longtable}
\usepackage{enumitem}
\usepackage{etoolbox}
\usepackage{float}
\usepackage{color}
\usepackage{wasysym}
\usepackage{hyperref}

\usepackage{mathtools}

\newcommand*{\DEBUG}{}%

\ifdefined\DEBUG
\newcommand{\fixme}[1]{{\textcolor{red}{\bf{\textsf{FIXME: #1}}}}}
\newcommand{\bug}[1]{{\textcolor{blue}{\bf{\textsf{BUG: #1}}}}}
\newcommand{\idea}[1]{{\textcolor{blue}{\bf{\textsf{IDEA: #1}}}}}

\newcommand{\TODO}[1]{{\textcolor{red}{\bf{\textsf{ 
TODO: #1
}}}}}
\else
\newcommand{\fixme}[1]{}
\newcommand{\bug}[1]{}
\newcommand{\TODO}[1]{}
\newcommand{\idea}[1]{}
\fi

\newclass{\COMSLIP}{COM\mbox{-}SLIP}
\newclass{\COMSLIPCUP}{COM\mbox{-}SLIP^{\cup}}

\newclass{\DCM}{DCM}
\newclass{\eDCM}{eDCM}
\newclass{\eNPDA}{eNPDA}
\newclass{\DPDA}{DPDA}
\newclass{\RDPDA}{RDPDA}
\newclass{\PDA}{PDA}
\newclass{\DCMNE}{DCM_{NE}}
\newclass{\TwoDCM}{2DCM}
\newclass{\NCM}{NCM}
\newclass{\eNCM}{eNCM}
\newclass{\eNQA}{eNQA}
\newclass{\eNSA}{eNSA}
\newclass{\eNPCM}{eNPCM}
\newclass{\eNQCM}{eNQCM}
\newclass{\eNSCM}{eNSCM}
\newclass{\DPCM}{DPCM}
\newclass{\NPCM}{NPCM}
\newclass{\NQCM}{NQCM}
\newclass{\NSCM}{NSCM}
\newclass{\NPDA}{NPDA}
\newclass{\TRE}{TRE}
\newclass{\NFA}{NFA}
\newclass{\DFA}{DFA}
\newclass{\NCA}{NCA}
\newclass{\DCA}{DCA}
\newclass{\DTM}{DTM}
\newclass{\NTM}{NTM}
\newclass{\DLOG}{DLOG}
\newclass{\CFG}{CFG}
\newclass{\ETOL}{ET0L}
\newclass{\EDTOL}{EDT0L}
\newclass{\CFP}{CFP}
\newclass{\ORDER}{O}
\newclass{\MATRIX}{M}
\newclass{\BD}{BD}
\newclass{\LB}{LB}
\newclass{\ALL}{ALL}
\newclass{\decLBD}{decLBD}
\newclass{\StLB}{StLB}
\newclass{\SBD}{SBD}
\newclass{\TCA}{TCA}

\newclass{\UFIN}{\LL(IND_{UFIN})}
\newclass{\UFINONE}{\LL(IND_{{UFIN}_1})}
\newclass{\FIN}{\LL(IND_{FIN})}
\newclass{\ILIN}{\LL(IND_{LIN})}
\newclass{\ETOLfin}{\LL(ET0L_{FIN})}

\newsavebox{\spacebox}
\begin{lrbox}{\spacebox}
\verb*! !
\end{lrbox}
\newcommand{\LL}{{\cal L}}

\newcommand{\N}{\mathbb{N}}

\begin{document}

\begin{frontmatter}

\title{On Finite-Index Indexed Grammars
and Their Restrictions
\tnoteref{t1}}

\tnotetext[t1]{\textcopyright 2022. This manuscript version is made available under the CC-BY-NC-ND 4.0 license \url{http://creativecommons.org/licenses/by-nc-nd/4.0/} The manuscript is published in F. D'Alessandro, O.H. Ibarra, I. McQuillan. On finite-index indexed grammars and their restrictions. {\it Information and Computation} 279, 104613 (2021).}

\author[label1]{Flavio D'Alessandro\fnref{fn1}\corref{corr}}
\address[label1]{Department of Mathematics\\ Sapienza University of Rome,
00185 Rome, Italy
 \\ and		\\
  Department of Mathematics, 
Bo\u gazi\c ci University \\
  34342 Bebek, Istanbul, Turkey}
\ead[label1]{dalessan@mat.uniroma1.it}
\cortext[corr]{Corresponding author}
\fntext[fn1]{Supported by EC-FP7 Marie-Curie/T\"UBITAK/Co-Funded Brain Circulation Scheme Project 2236 (Flavio D'Alessandro).}

\author[label2]{Oscar H. Ibarra\fnref{fn2}}
\address[label2]{Department of Computer Science\\ University of California, Santa Barbara, CA 93106, USA}
\ead[label2]{ibarra@cs.ucsb.edu}
\fntext[fn2]{Supported, in part, by
NSF Grant CCF-1117708 (Oscar H. Ibarra).}

\author[label3]{Ian McQuillan\fnref{fn3}}
\address[label3]{Department of Computer Science, University of Saskatchewan\\
Saskatoon, SK S7N 5A9, Canada}
\ead[label3]{mcquillan@cs.usask.ca}
\fntext[fn3]{Supported, in part, by Natural Sciences and Engineering Research Council of Canada Grant 2016-06172 (Ian McQuillan).}

\begin{abstract}
The family, $\ILIN$, of languages generated
by linear indexed grammars has been 
studied in the literature.
It is known that the Parikh image of every language in $\ILIN$
is semi-linear.  However, there are bounded semi-linear languages that
are not in $\ILIN$.  Here, we look at larger families of (restricted)
indexed languages and study their combinatorial and decidability properties, and their relationships.
\end{abstract}


\end{frontmatter}

\section{Introduction}
\label{sec:intro}

Indexed grammars \cite{A,A1} are a natural generalization of context-free grammars, where variables keep stacks of indices.
Although they are included in the context-sensitive languages, 
the languages generated by indexed grammars are quite broad as they contain some non semi-linear languages.
Several restrictions have been studied that have desirable computational properties.
Linear indexed grammars were first created, restricting the number of variables on the right hand side to be at most one \cite{DP}.
Other restrictions include another system named exactly linear indexed grammars \cite{Gazdar1988} (see also \cite{Vijay-Shanker1994}), which are different than
the first formalism, although both are sufficiently restricted to only generate semi-linear languages. In this
paper, we only examine the first formalism of linear indexed grammars.

We study indexed grammars that are restricted to be
finite-index, which is a generalization of linear indexed grammars \cite{DP}. 
Such grammars generate languages that inherit several properties satisfied by context-free languages $\CFL$.
Grammar systems that are {\em $k$-index}
are restricted so that, for every word generated by the grammar,
there is some successful derivation where at most $k$ variables (or nonterminals) appear in 
every sentential form of the derivation \cite{DP1989,C3.6,RS,RozenbergFiniteIndexGrammars}. A system is finite-index if it is $k$-index for some $k$.
It has been found that that when restricting many different types of grammar systems to be finite-index, their languages coincide. This is the case for finite-index $\ETOL$, $\EDTOL$, context-free programmed
grammars, ordered grammars, and matrix context-free grammars.

We introduce the family,
$\FIN$, of languages generated by finite-index
indexed grammars and a sub-family, $\UFIN$,
of languages generated by uncontrolled finite-index
indexed grammars, where every successful derivation has to be
finite-index. 
The grammars generating the languages of $\UFIN$ have been very recently studied under the name of
{\em breadth-bounded grammars}, and it was shown that
this family is a semi-linear full trio. We also study a special case of the latter, called
$\UFINONE$ that restricts branching productions. 
We then show the following:
\begin{enumerate}
\item
All families are full trios.

\item The semi-linearity property of $\UFIN$ and $\UFINONE$ is extended to a bigger family, showing, more generally,  that, 
if $\cal C$ is an arbitrary full trio of semi-linear languages and $\LL(\NCM)$ is the family of languages accepted by 
one-way deterministic reversal-bounded multicounter machines, then every language in the family 
$$  \{L_1 \cap L_2  : L_1 \in {\cal C}, L_2 \in \LL(\NCM)\},$$
 has a semi-linear Parikh image.
%
%

 \item The following conditions are equivalent for a bounded language $L$:
\begin{itemize}
\item $L \in \UFINONE$,
\item $L \in \UFIN$,
\item $L$ is bounded semi-linear,
\item $L$ can be generated by a finite-index $\ETOL$ system,
\item $L$ can be accepted by a DFA augmented with reversal-bounded counters,
\end{itemize}
\item Every finite-index $\ETOL$ language is in $\FIN$,
\item
$\CFL\ \subset \  \ILIN\ \subset \UFINONE  \subseteq \UFIN \subset \FIN$, 

\item Containment and equality are decidable for bounded languages in $\ILIN$ and $\UFIN$.
\end{enumerate}

\section{Preliminaries}
\label{sec:prelims}

We assume a basic background in formal languages and automata theory \cite{Be,Gins,C3.6,C4.8}.

Let $k$ be a positive integer and let $\N^k$ be the additive free commutative monoid of $k$-tuples of non negative integers. If $B$ is a subset of $\N^k$, $B^\oplus$ denotes the submonoid of 
  $\N^k$ generated by $B$.

An {\em alphabet} is a finite set of symbols, and given an alphabet $A$,
$A^*$ is the free monoid generated by $A$. An element $w \in A^*$
is called a {\em word}, the empty word is denoted by $\lambda$, 
and any $L \subseteq A^*$ is a {\em language}. 
The {\em length} of
a word $w\in A^*$ is denoted by $|w|$, and the number of $a$'s, $a \in A$, in $w$
is denoted by $|w|_a$, extended to subsets $X$ of $A$ by 
$|w|_X = \sum_{a \in X} |w|_a$.

Let $A = \{a_1, \ldots, a_t\}$ be an alphabet of $t$ letters, and let $\psi: A^* \rightarrow \N ^t$ be the corresponding {\em Parikh morphism}
defined by $\psi(w) = (|w|_{a_1}, \ldots, |w|_{a_t})$.

A set $B \subseteq \N^k$ is a {\em linear set} if there exist
vectors ${\bf b}_0, {\bf b}_1 , \ldots, {\bf b}_{n}$ of $\N^k$ such that 
 $B = {\bf b}_0 + \{{\bf b}_1 , \ldots, {\bf b}_{n} \}^\oplus.$
 Further, $B$ is
called a {\em semi-linear set} if $B = \bigcup _{i=1} ^m B_i,  m\geq 1,$
for linear sets $B_1, \ldots, B_m$.
A language $L \subseteq A^*$ is said to be {\em semi-linear} if the
Parikh morphism applied to $L$ gives a semi-linear set. A language family
is said to be semi-linear if all languages in the family are semi-linear.
Many known families are semi-linear, such as the regular languages and 
context-free languages  (denoted by $\CFL$, see \cite{Be,Gins,C3.6,C4.8}),
and finite-index $\ETOL$ languages $\ETOLfin$, see \cite{RS,RozenbergFiniteIndexETOL}).

A language $L$ is termed {\em  bounded} if there exist non-empty 
 words $u_1, \dots, u_k$, with $k\geq 1$,  such that $L \subseteq u_1 ^*\cdots u_k^*$.
  Let $\varphi : \N ^k \rightarrow  u_1^*\cdots u_k^*$
be the map defined as: for every tuple $(\ell_1, \ldots, \ell_k)\in \N ^k$, $$\varphi  (\ell_1, \ldots, \ell_k) = u_1 ^ {\ell_1}\cdots u_k^ {\ell_k}.$$
The map $\varphi$ is called the {\em Ginsburg map}.
   \begin{definition}\label{+}
  A bounded  language $L\subseteq u_1^* \cdots u_k^*$ is said to be {\em bounded Ginsburg semi-linear}  if  there exists a semi-linear set 
  $B$ of $\N^k$  such
that $\varphi (B) = L$.
\end{definition}
In the literature, bounded Ginsburg semi-linear has also been called just bounded semi-linear, but we will use the terminology bounded Ginsburg semi-linear henceforth in this paper.

A {\em full trio} is a language family closed under morphism,
inverse morphism, and intersection with regular languages \cite{Be}.

We will also relate our results to the languages accepted by
one-way nondeterministic reversal-bounded multicounter machines
(denoted by $\LL(\NCM)$), and to
one-way deterministic reversal-bounded multicounter machines
(denoted by $\LL(\DCM)$. 
These are NFAs (DFAs) augmented by a set of counters that can switch between increasing and decreasing a fixed number of times \cite{BB1974,Ibarra1978}).

\section{Restrictions on Indexed Grammars}
\label{sec:indexed}

We first recall the definition of indexed grammar introduced in  \cite{A}  by following \cite{C4.8},   Section 14.3 (see also \cite{DP1989} for a reference book for grammars). 
 \begin{definition}\label{def0}
  An indexed grammar is a 5-tuple $G = (V, T, I, P, S)$, where 
\begin{itemize}
\item $V, T, I$ are finite pairwise disjoint sets: the set of variables, terminals, and indices, respectively;
\item $P$ is a finite set of  productions of the forms
$$
\mbox{\bf 1)\ } A \rightarrow \nu, \quad \mbox{\bf 2)\ } A \rightarrow Bf, \quad \mbox{or}\quad 
\mbox{\bf 3)\ } Af \rightarrow \nu,$$
where $A, B\in V,$ $f\in I$ and $\nu \in (V\cup T)^*$;
 \item $S\in V$ is the start variable.
\end{itemize}
\end{definition}
 Let us now define the derivation relation $\Rightarrow_G$ of $G$. Let $\nu$ be an arbitrary sentential form of $G$,
 $$u_1 A_1 \alpha_1 u_2 A_2 \alpha_2 \cdots u_k A_k \alpha_k u_ {k+1},$$ 
 with $A_i\in V, \alpha _i \in I^*, u_i\in T^*.$
 For a sentential form $\nu ' \in  (VI^* \cup T)^*$, we set $\nu \Rightarrow _G \nu '$ if one of the following three conditions holds:
\begin{description}
\item[ 1)] In $P$, there exists a production  of the form (1) $A \rightarrow w_1 C_1  \cdots w_{\ell} C_\ell w_ {\ell+1}$, $C_j \in V,
w_j \in T^*$, 
 such that in the sentential form $\nu$, for some  $i$ with $1\leq i\leq k,$ one has $A_i =A$ and 
$$\nu' = u_1 A_1 \alpha_1 \cdots u_{i} (w_1 C_1\alpha_i  \cdots w_{\ell} C_\ell \alpha_iw_ {\ell+1}) u_{i+1} A_{i+1}\alpha_{i+1} \cdots u_k A_k \alpha_k u_ {k+1}.$$

\item[ 2)] In $P$,  there exists a production  of the form (2)  $A\rightarrow Bf$ such that in the sentential form $\nu$, 
for some  $i$ with $1\leq i\leq k,$ one has $A_i =A$  and 
$\nu '  = u_1 A_1 \alpha_1 \cdots u_{i} (B f \alpha_i) u_{i+1} A_{i+1}\alpha_{i+1} \cdots u_k A_k \alpha_k u_ {k+1}.$

\item[ 3)] In $P$, there exists a production  of the form (3) $A f\rightarrow w_1 C_1  \cdots w_{\ell} C_\ell w_ {\ell + 1}$, $C_j \in V,
w_j \in T^*$,  
 such that in the sentential form $\nu$, for some  $i$ with $1\leq i\leq k,$ one has $A_i =A$, $\alpha_i = f\alpha'_i, \alpha'_i\in I^*$,  and
$$\nu  ' = u_1 A_1 \alpha_1 \cdots u_{i} (w_1 C_1\alpha'_i  \cdots w_{\ell} C_\ell \alpha'_iw_ {\ell+1}) u_{i+1} A_{i+1}\alpha_{i+1}  \cdots u_k A_k \alpha_k u_ {k+1}.$$
In this case, one says that the index $f$ is consumed.
 \end{description}

 For every $n\in \N$, $\Rightarrow_G ^n$ stands for the $n$-fold product of $\Rightarrow_G$ and $\Rightarrow_G ^*$ stands for the reflexive and transitive 
closure of $\Rightarrow_G$. The language $L(G)$ generated by  $G$ is the set $L(G) = \{u\in T^* : S \Rightarrow_G ^* u\}.$
 \bigskip

\noindent
{\bf Notation and Convention.}
{\em  In the sequel we will adopt the following notation and conventions for an indexed grammar $G$.
\begin{itemize}
   \item  If no ambiguity arises, the relations 
  $\Rightarrow_G, $ $\Rightarrow_G ^n, \ n\in \N$, and $\Rightarrow_G ^*$ will be simply denoted by  $\Rightarrow$,
  $\Rightarrow^n$, and $\Rightarrow^*$, respectively. 
 \item capital letters as $A, B,...etc$ (as well as its indexed variant)  will  denote variables of $G$. 
 
 \item the small letters $e, f$, as well as $f_i$, will be used to denote
 indices while $\alpha$, $\beta$ and $\gamma$, as well as its indexed variant (as for instance $\alpha_i$), will denote arbitrary words over $I$.
 \item Small letters as $a, b, c, ... etc$ (as well as its indexed variant) will  denote   letters of $T$ and 
  small letters as $u, v, w,r..., etc$ (as well as its indexed variant) will  denote  words over $T$.
  \item $\nu$ and $\mu$, as well as $\nu_i$ and $\mu_i$,  will denote arbitrary sentential forms of $G$.
 \item in order to shorten the notation, according to Definition \ref{def0}, if $p$ is a production of $G$ of the form (1) or (3), we will simply write
 $$Af \rightarrow \nu, \quad f\in I \cup \{\lambda\},$$
where it is understood that if $f=\lambda$,  the production $p$ has form (1) and if $f\in I$, the production $p$ has form (3). 

\item  \medskip

\noindent
If $ p \in P$ is a production of $G$,  then $\ \mu \Rightarrow _{p} \nu$ denotes  the $1$-step derivation of $G$ defined by $p$;

\item  \medskip

\noindent
  a derivation of $G$ of the form
$\ \nu_0 \Rightarrow _{p_1} \nu_1 \Rightarrow _{p_2} \cdots  \Rightarrow _{p_{n}} \nu_{n}$ will be also
shortly denoted as $\ \nu_0 \Rightarrow _ {p_1 \cdots p_n}  \nu_n$.

\end{itemize} }

The following set of definitions defines the main objects studied in this paper.
Let $G$ be an indexed grammar and let $L(G)$ be the language generated by $G$. 
The first definition is from \cite{DP}.
 \begin{definition}\label{eq-def-1}
We say that $G$ is  {\em linear} if  the right side component of every production of $G$ has at most one variable. 
A language $L$  is said to be {\em linear indexed} if there exists a linear indexed grammar $G$ such that $L=L(G)$.
\end{definition}

 \begin{definition}\label{eq-def-2}
Given an integer $k\geq 1$, a derivation  
$\nu_0 \Rightarrow \nu_1 \Rightarrow \cdots \Rightarrow \nu_n$  of $G = (V,T,I,P,S)$, is said to be of 
{\em index-$k$} if $|\nu_i|_V \leq k$, for all $i$, $0 \leq i \leq n$.
\end{definition}
 \begin{definition}\label{eq-def-3}
Given an integer $k\geq 1$,
$G$ is said to be {\em of index-$k$} if, for every word $u\in L(G)$,
there exists a derivation of $u$ in $G$ of index-$k$.

A language $L$ is said to be an {\em indexed language of index-$k$} if there exists an indexed
 grammar $G$ of index-$k$ such that $L=L(G)$.  
  An indexed language $L$ is said to be {\em  of finite-index} if $L$ is of index-$k$, for some $k$. 
 \end{definition}
 \begin{definition}\label{eq-def-4}
An indexed grammar  $G$ is said to be {\em uncontrolled} index-$k$ if,
for every derivation $\nu_0 \Rightarrow \cdots \Rightarrow \nu_n$ generating
$u \in L(G), |\nu_i|_V \leq k$, for all $i$, $0 \leq i \leq n$. $G$ is
uncontrolled finite-index if $G$ is uncontrolled index-$k$, for some $k$.
A language $L$ is said to be an {\em uncontrolled finite-index indexed language} if there exists an uncontrolled finite-index
 grammar $G$ such that $L=L(G)$.
  \end{definition}
  \begin{remark}
  It is worth noticing that, according to Definition \ref{eq-def-3}, if $G$ is a grammar of index-$k_1$, then $G$ is a grammar of index-$k_2$, for every
  integer $k_1\leq k_2$.
  \end{remark}

    \begin{remark}
   It is interesting to observe that Definition \ref{eq-def-4} corresponds, in the case of context-free grammars, 
  to the definition of  {nonterminal bounded grammar} (cf \cite{C3.6},   Section 5.7).  We recall that nonterminal bounded grammars
   are equivalent to {ultralinear grammars} and thus provide a characterisation of the family of languages that are accepted by Finite-Turn pushdown automata. 
 \end{remark}
   Finally let us denote by
 \begin{itemize}
 \item $\FIN$ the family of finite-index indexed languages;
 \item $\UFIN$ the family of uncontrolled finite-index indexed languages;
\item $\ILIN$ the family of linear indexed languages.

\end{itemize}

Uncontrolled finite-index grammars have been studied under the name of  
breadth-bounded indexed grammars in \cite{ICALP2015}, where the following result has been proved.
 
\begin{theorem}
$\UFIN$ is a semi-linear full trio.
\end{theorem}
(Makoto Kanazawa has pointed out to the authors that Georg Zetzsche's result -- the Parikh image of every language in $\UFIN$ is semilinear -- 
can be also obtained as corollary of a result proved in his paper \cite{K2014}).

The family  $\ILIN$
 has been  introduced  in \cite{DP} where results of algebraic and combinatorial
nature characterize the structure of its languages. Recall that a linear indexed grammar $G$ is said to be 
{\em right linear indexed} if, according to Definition \ref{def0}, in every production $p$ of $G$ of the form (1) 
or (3), the right hand component $\nu$ of $p$  has the form $\nu = u$, or $\nu = uB$, where $u\in T^*, B\in V$. 
In \cite{A} (see also \cite{DP}), the following theorem has been proved:
 \begin{theorem} \label{thm-A} If $L$ is an arbitrary language, $L$ is context-free if and only if there exists 
 a right linear indexed grammar $G$ such that $L=L(G)$.
\end{theorem}
From this, the following is evident.
\begin{theorem}\label{First-hierarchy}
\label{initialhierarchy}
$\CFL\ \subset \  \ILIN \ \subset \  \UFIN \  \subseteq \ \FIN.$
 \end{theorem}
 Indeed Theorem \ref{thm-A} provides the inclusion $\CFL \subseteq  \ILIN$.  
 The inclusions 
$\ILIN\subseteq \UFIN \subseteq \FIN$ come immediately from the definitions of the  corresponding families.
In \cite{DP} (see Theorem 2.8), it is shown that for an alphabet $T$,
and a letter $\$\notin T$,  if $M_1\$M_2$, with $M_1, M_2 \subseteq T^*$, is a linear
indexed language, then $M_1$ or $M_2$ is a context-free language.
 Let $T$ be an alphabet with at least two letters.
 Let $L_1 =  \{u^2: u\in T^*\}$, and let $L_2 =  \{u^2\$v^2: u, v \in T^*\}$. One easily sees that $L_1 \in \ILIN \setminus \CFL$,
and, since $L_2 = L_1 \$ L_1$, by the previous remark, $L_2\notin \ILIN$.  
 On the other hand, it is easily shown that $L_2 \in \UFIN$. 
More generally, one can verify that, for every $k\geq 1$, $L_k =  \{u^k\$v^k: u, v \in T^*\}\in    \UFIN$.  
 
By applying the same argument, one has that, on the alphabet $T = \{a, b, c, \$\}$, the language  $L =  \{a^n b^n c^n \$ a^m b^m c^m : n,m \geq 0\}$
cannot be linear indexed. 

  Next, closure under union and product is addressed for the family $\FIN$.
 \begin{lemma}\label{lm-0} The family $\FIN$ is closed under union and concatenation. 
\end{lemma}
\begin{proof}
Let $L_1$ and $L_2$ be indexed languages of indices  $k_1$ and $k_2$ respectively, and let $G_1$ and $G_2$ be grammars
 $$G_1 = (V_1, T_1, I_1, P_1, S_1), \quad G_2 = (V_2, T_2, I_2, P_2, S_2),$$
 such that $L_1= L(G_1)$ and $L_2=L(G_2)$. 
 Since we may rename variables and indices without changing the language generated, we assume
 that $V_1\cap V_2 = I_1 \cap I_2= \emptyset$. Moreover let $S$ be a new variable not in $V_1\cup V_2$. 
 
 Construct a new grammar $G = (V, T, I, P, S)$, where
 $V = V_1 \cup V_2 \cup \{S\}$, $I = I_1 \cup I_2$, and $P$ is equal to
 $P_1 \cup P_2$, plus the two productions $S\rightarrow S_1$ and $S\rightarrow S_2$.
  It is easily checked that $L_1 \cup L_2 = L(G)$ and $G$ is of index $\max \{k_1, k_2\}$.

For concatenation, let $G' = (V, T, I, P', S)$ be the grammar obtained from $G$, by setting $P'$ equal to
 $P_1 \cup P_2$, plus the production $S\rightarrow S_1 S_2$. 
 It is easily checked that $L_1L_2 = L(G')$ and $G'$ is of index $1+ \max \{k_1, k_2\}$.
  \qed
\end{proof}

Next, we show that $\FIN$ is a full trio.
As a consequence, by using Nivat's theorem for the characterisation of rational relations of free monoids (see \cite{Be}, Ch. III, Thm 4.1), 
we will prove the fact that they are closed under rational transductions.
The proof is structured using a chain of  lemmas.
  \begin{lemma}\label{lm:1}  
$\FIN$ is closed under morphisms. 
\end{lemma}
\begin{proof}
 
Let $L \in  \FIN$ and let $G = (V, T, I, P, S)$ be a $k$-index indexed grammar such that $L=L(G)$.
Let $\varphi : T^* \rightarrow (T')^*$ be a morphism where $T$ and $T'$ are two alphabets. Construct a new grammar $G'$
by replacing each production of $G$ of the form $$Xf \rightarrow u_1 X_1  \cdots u_\ell X_\ell u_{\ell +1},$$ where
$f\in I \cup \{\lambda\}, \ u_i \in T^*, \ X, X_i \in V,$ by the production
$$Xf \rightarrow \varphi(u_1) X_1  \cdots \varphi (u_\ell) X_\ell \varphi (u_{\ell +1}).$$
It is easily verified that the resulting grammar $G'$ satisfies $\varphi (L) = L(G')$ and $G'$ is a $k$-index grammar. 
\qed\end{proof}

%

  \begin{lemma}\label{lm:3}  
$\FIN$ is closed under intersection with regular languages. 
\end{lemma}

\begin{proof} 
 Let  $G = (V, T, I, P, S)$ be a finite-index  indexed grammar and let  $L=L(G)$.
Let $R$ be a regular language and let $\cal A$ be the finite automaton accepting $R$. 
Our main goal is to construct a finite index grammar $G'$ such that, given an arbitrary word $w\in L \cap R$, on one hand, generates $w$, and,
on the other hand, simulates in the automaton  $\cal A$ the recognition process of $w$. It is thus convenient to
rewrite the productions of the grammar $G$ in a suitable canonical form. 
Indeed, in the construction of the grammar $G'$, such form allows then to match the $1$-step derivations of $G$
with the atomic transitions of $\cal A$. For this reason, the 
following claim is needed. 
\medskip

%
%

\noindent
{\bf Claim 1.} {\em
There exists a finite-index indexed grammar 
$G' = (V', T, I', P', S')$ generating $L$ such that $I'=I$ and the productions of $P'$ are of the form: 
 $$
\mbox{1)\ } A \rightarrow \nu, \quad \mbox{2)\ } A \rightarrow Bf, \quad \mbox{or}\quad 
\mbox{3)\ } Af \rightarrow \nu,$$
where $A, B\in V',$ $f\in I'$ and $\nu  \in (V'\cup T)^*$ is a word of the form
$$\nu=u, \quad \mbox{or} \quad \nu =  uXZ, \quad \mbox{or} \quad \nu = uXv, \quad \quad X, Z \in V', \quad u, v\in T^*.$$}
 \medskip

 \noindent
{\bf Proof of the Claim.}
 Let us  first assume that $G$ has a sole production $p$  of the form 
\begin{equation}\label{eq:reducedgramm} 
A \rightarrow \nu= u_1 X_1 u_2 X_2 \cdots u_k X_k u_{k+1},  \quad k\geq 2,\  A, X_i\in V, u_i\in T^*.
\end{equation}
Define the following list of productions:
  \begin{enumerate}
\item[i.] $A \rightarrow u_1 X_1 Z_1$

\item[ii.] For every $j=1,\ldots, k-2, \quad Z_j \rightarrow u_{j+1} X_{j+1} Z_{j+1}$

\item[iii.] $Z_{k-1}  \rightarrow u_k X_k u_{k+1},$
\end{enumerate}
where $Z_j, (j =1, \ldots, k-1)$, are new variables not in $V.$ 
 
Remove the production (\ref{eq:reducedgramm}) from $P$, add to $P$ the list of productions defined at (i)-(ii)-(iii) above, and add to $V$ the 
corresponding list of new variables $Z_j$'s. Let $G'$ be the grammar obtained from $G$ by using the previous transformation.
 We now observe that $G'$ satisfies the claim and that the derivation of $G$ defined by (\ref{eq:reducedgramm}) is simulated by the 
 derivation of $G'$:
$$A \Rightarrow _{G'} u_1 X_1 Z_1 \Rightarrow _{G'} u_1 X_1 u_2 X_2 Z_2  \Rightarrow _{G'} \cdots  
\Rightarrow _{G'} u_1 X_1 u_2 \cdots u_{k-1} X_{k-1} Z_{k} 
\Rightarrow _{G'} \nu.$$
Moreover  such derivation has index not larger than that of $G$.  
From the latter remark, it is easily checked, by induction on the length of the derivations of $G'$, that $G'$ has the same index of $G$ and that $L = L(G')$.

The case of productions $Af \rightarrow \nu, f\in I$, is 
similarly treated. If $G$ has two or more productions of the forms previously considered, the claim is obtained by
iterating the previous argument.

\bigskip
Let $G = (V, I, T, P, S)$ be a finite-index indexed grammar in the form given by the previous Claim.
Let $R$ be a regular language over $T$ and let  ${\cal A} = (Q, T, \tau, q_0, K)$  be a finite deterministic and complete automaton accepting $R$, where 
$Q$ is the set of states of $\cal A$, $\tau: Q\times T\rightarrow Q$ is its transition function, 
$q_0 \in Q$ is its unique initial state while $K$ is the set of final states of $\cal A$. In the sequel, for the sake of
simplicity,  the extension of the function $\tau$ to the set $Q\times T^*$ will be still denoted  by $\tau$.

\noindent
We proceed to  construct a new finite-index 
grammar $G'$ such that $G'= (V', I', T, P', S')$ and $L(G') = L \cap R$. 

The set $V'$ of variables of $G'$ will be of the form $\langle p, X, q \rangle$, where $p$ and $q$ are in $Q$ and $X$ is in $V$,
together with a new symbol $S'$, denoting the start variable of $G'$. 

The set $I'$ of indices of $G'$ is a copy of $I$ disjoint with it. For every index $f$ of $I$, we will denote by $f'$ the  corresponding copy of $f$ in $I'$
 (it is understood that if $f=\lambda$ then $f'=\lambda$). 

The set $P'$ of productions of $G'$ is defined as follows. 
\begin{enumerate}
\item If $Af\rightarrow u$ is in $P$, where $f\in I \cup \{\lambda\}, u \in T^*$, and $\tau (p, u)=q$, then $P'$ contains the set of productions
$\langle p, A, q \rangle f' \rightarrow u$, for all $p, q\in Q$ such that $p$ is transformed to $q$ by reading $u$.

\item If $A\rightarrow Bf$ is in $P$, where $f\in I$, then $P'$ contains the set of productions
$$\langle p, A, q \rangle \rightarrow \langle p, B, q \rangle f',$$ where    $p, q$ are two
arbitrary states of $Q$.

\item If $Af\rightarrow vDw$ is in $P$, where $f\in I \cup \{\lambda\}, A,D \in V, v,w \in T^*$, then $P'$ contains, for all $p, q, r, s\in Q$, the set of productions
$$\langle p, A, q \rangle f' \rightarrow v\langle r, D, s \rangle w,$$
provided that $\tau (p, v)=r$, and $\tau (s, w)=q$.

\item If $Af\rightarrow uBC$ is in $P$, where $f\in I \cup \{\lambda\}, A,B,C \in V, u \in T^*$, then $P'$ contains, for all $p, q, r', r''\in Q$, the set of productions
$$\langle p, A, q \rangle f' \rightarrow     u \langle r', B, r'' \rangle \langle r'', C, q \rangle, $$
provided that $\tau (p, u)=r'$.

\item Finally $P'$ contains the production $S' \rightarrow \langle s_0, S, p \rangle$, for all $p\in K$.
\end{enumerate}
No other production different from the form specified in the list above is in $P'$.

The first task is to show that $L\cap R = L(G')$. For this purpose, we first show that:
  $\langle p, A, q \rangle f'_1 \cdots f'_i \Rightarrow _{G'} ^* u,$ with $ i\geq 0$, $u\in T^*$, if and only if 
$ A  f_1 \cdots f_i \Rightarrow _{G} ^* u$ and $\tau (p, u)=q$. 
Indeed, from this statement, we get $S' \Rightarrow _{G'}  \langle s_0, S, q  \rangle   \Rightarrow _{G'} ^* u$, for some $q\in K$, if and only if
 $S \Rightarrow _{G} ^* u,$ and $\tau (s_0, u)=q$, which is sufficient to complete the proof.
  
Let us first prove that:
\[
\begin{array}{rl}
\mbox{(*)}  &  \mbox{If  $\langle p, A, q \rangle f'_1 \cdots f'_i \Rightarrow _{G'} ^\ell u$ is a derivation of $G'$ of length $\ell\geq 0$ then}    \\
  &       \mbox{$ A  f_1 \cdots f_i \Rightarrow _{G} ^* u$ and $\tau (p, u)=q$.}
\end{array}
\]
(*) is easily checked to be true for derivations of length $1$. Now suppose that (*) is true for all $m < \ell$ with $m\geq 1$ and let 
$\langle p, A, q \rangle f'_1 \cdots f'_i \Rightarrow _{G'} ^\ell u$ be a derivation of $G'$ of length $\ell$. Such a derivation can be of one of the following forms. 
\begin{description}
 \item[ (i)] 
$\langle p, A, q \rangle f'_1 \cdots f'_i \Rightarrow _{G'} \langle p, B, q \rangle f' f'_1 \cdots f'_i \Rightarrow^{\ell -1} _{G'}  u, \quad f'\in I',$ 
\medskip

\noindent
that is, the first production of the derivation has the form $(2)$. By the inductive hypothesis, we then have
$Bf f_1 \cdots f_i \Rightarrow ^*  _{G} u$ and $\tau (p, u) = q$, which yields 
$A f_1 \cdots f_i \Rightarrow_G Bf f_1 \cdots f_i \Rightarrow ^*  _{G} u$ and $\tau (p, u) = q$.
\medskip

\noindent
\item[ (ii)] 
$\langle p, A, q \rangle f'  f'_1 \cdots f'_i  \Rightarrow _{G'} v\langle r, D, s  \rangle   f'_1 \cdots f'_i  w\Rightarrow^{\ell -1} _{G'}  u, \, f'\in I' \cup \{\lambda\},$ 
\medskip

\noindent
that is, the first production of the derivation has the form $(3)$. Set $u = vu'w$. From the latter, we get
$\langle r, D, s  \rangle   f'_1 \cdots f'_i  \Rightarrow^{\ell -1} _{G'}  u'$ so that, 
by the inductive hypothesis,  
$D f_1 \cdots f_i \Rightarrow ^*  _{G} u'$ and $\tau (r, u') = s$. 
On the other hand, we know that 
$$Af \Rightarrow _{G} vDw, \quad \quad \tau (p, v)=r, \quad \tau (s,w)=q,$$ thus yielding
 $Af f_1 \cdots f_i \Rightarrow   _{G} vDf_1 \cdots f_iw \Rightarrow ^*  _{G} vu'w=u$. Furthermore, $\tau (p, v) = r$, $\tau (s, w) = q$
which gives  $\tau (p, u) = q$.
\medskip

\noindent

\item[ (iii)] 
$\langle p, A, q \rangle f'  f'_1 \cdots f'_i  \Rightarrow _{G'} v    \langle r', B, r''  \rangle f'_1 \cdots f'_i \langle r'', C, q  \rangle    
f'_1 \cdots f'_i  \Rightarrow^{\ell -1} _{G'}  u,$ \\ $f'\in I' \cup \{\lambda\}, \quad  r'=\tau (p, v),$ 
\medskip

\noindent
that is, the first production of the derivation has the form $(4)$. Set $u = vu'$, with $u'\in A^*$. From the second sentential form, we get
$$ \langle r', B, r''  \rangle f'_1 \cdots f'_i  \Rightarrow^{\ell_1} _{G'}  u'_1, \quad
 \langle r'', C, q  \rangle f'_1 \cdots f'_i  \Rightarrow^{\ell_2} _{G'}  u'_2,$$ where $u'=u'_1 u'_2,$ with $u' _1, u' _2\in A^*, l_1<l, l_2<l$.
By the inductive hypothesis,  we have 
$$
B f_1 \cdots f_i \Rightarrow ^*  _{G} u'_1, \quad C f_1 \cdots f_i \Rightarrow ^*  _{G} u'_2,
$$
together with \begin{equation}\label{eq:lm:3}\tau (r', u'_1) = r'', \quad \tau (r'', u'_2) = q,\end{equation}  thus yielding
\medskip

\noindent
$Af f_1 \cdots f_i \,\Rightarrow  _{G} \, vBf_1 \cdots f_i Cf_1 \cdots f_i \, \Rightarrow ^*  _{G}\, $
%
$ vu'_1 Cf_1 \cdots f_i  \,  \Rightarrow ^*  _{G} \, vu'_1u'_2=vu'=u.$ 
 \medskip

\noindent
Finally, from (\ref{eq:lm:3}) and $\tau (p, v) = r'$,  
we get  $\tau (p, u) = q$.
\medskip

\noindent
\item[ (iv)] 
$\langle p, A, q \rangle f' f'_1 \cdots f'_i \Rightarrow _{G'} w f'_1 \cdots f'_i \Rightarrow^{\ell -1} _{G'}  u,$ 
\medskip

\noindent
that is, the first production of the derivation has the form $(1)$. 
In this case, $f'_1 = \cdots = f'_i=\lambda,$ and $\ell=1$ so that the claim is trivially proved.\end{description}

Since the latter  cases represent all the possible ways an arbitrary derivation can start,  (*) is proved. Similarly, 
taking into account the fact that the productions of $G$ are in the form given in Claim 1, 
one proves by induction on the length of a derivation in $G$ that if 
 $A f_1 \cdots f_i \Rightarrow _{G} ^\ell u$ is a derivation of $G$ of length $\ell\geq 0$ and $\tau (p, u)=q$ then
$\langle p, A, q \rangle f'_1 \cdots f'_i \Rightarrow _{G'} ^* u$. By the previous remark, this implies that $L(G') = L(G) \cap R$. 
 
 Let $\delta '$ be a derivation of $G'$. By induction on the length of $\delta'$, one can prove the existence of a derivation
 $\delta$ of $G$ that simulates (step by step) $\delta'$. This implies that if $G$ is a grammar of finite index,
 then $G'$ is of the same type as well. 
  This concludes the proof.
     \qed\end{proof}

Next, we show closure under an inverse morphism. 

Let $T$ and $T'$ be two alphabets with $T\subseteq T'$ and let $\widehat{\pi_{T}} : (T')^* \rightarrow T^*$ be the 
projection of $(T')^*$ onto $T^*$, that is the epi-morphism from $(T')^*$ onto $T^*$ generated by the mapping
${\pi_{T}} : T' \rightarrow T \cup \{\lambda\}$ 
$$\forall \ \sigma \in T',  \pi_{T} (\sigma) = 
\left \{
\begin{array}{ccc}
 \lambda  &  \ \mbox{if}\ \sigma \notin T,  \\
  \sigma  &  \ \mbox{if}\ \sigma \in T
\end{array} .
\right.
$$
In the sequel, for the sake of simplicity, we denote the projection $\widehat{\pi_{T}}$ by ${\pi_{T}}$. 
It is useful to remark that, for every $w\in T^*$ and $w'\in (T')^*$, with $w=a_1\cdots a_n, \, n\geq 0, \, a_i \in T$, 
\begin{equation}\label{eqproj}
w'\in \pi_T ^{-1} (w) \ \Leftrightarrow\ w' = w_1 a_1 \cdots w_{n}a_n w_{n+1}, \ w_i \in (T'\setminus T)^*.
\end{equation}

 \begin{lemma}\label{lm:2}  
   If $L\in \FIN$ with $L \subseteq T^*$, then $\pi_T ^{-1} (L) \in \FIN$. 
\end{lemma}
\begin{proof} 
  Let $G = (V, T, I, P, S)$ be a  finite-index indexed grammar generating $L$. We construct a   finite-index grammar $G'$
 generating $\pi_T ^{-1} (L)$ with the same index.

 For this purpose, let $p=Xf \rightarrow \nu$, with $X \in V, f\in I \cup \{\lambda\}$, and $\nu \in (VI^* \cup T^*)^*$,  be a production  of $G$ of the form (1) or (3)
(according to Definition \ref{def0}). Then  $p$ has  the form
 $$Xf \rightarrow \nu= u_1 X_1 \cdots u_k X_k u_{k+1}, \quad u_i \in T^*,$$
 where $ X, X_i\in V,$ with $i=1,\ldots, k,$ and, for every $i=1,\ldots, k+1$, 
$$u_i = a_{i,1} \cdots a_{i, n_i}, \quad   n_i \geq 0, \quad a_{i, j} \in T.$$
 Let us associate with $p$, the following set of productions:
\begin{itemize}
\item $X f \rightarrow Y_{1,0}\cdots Y_{k,0},$
\item $\forall \ i=1,\ldots, k,\,  \forall \ j=0,\ldots, n_i , \quad Y_{i, j} \rightarrow cY_{i, j}, \quad c\in T' \setminus T,$
\item $\forall \ i=1,\ldots, k,\,  \forall \ j=0,\ldots, n_i -1, \quad  Y_{i,j} \rightarrow a_{i, j+1}Y_{i,j+1},$ 
  \item $\forall \ i=1,\ldots, k-1, \quad Y_{i,n_i} \rightarrow X_i, \quad\quad Y_{k,n_k} \rightarrow Y' _{k,0},$
 \item $\forall \ j=0,\ldots, n_{k+1} , \quad Y' _{k, j} \rightarrow Y'_{k, j}c, \quad c\in T' \setminus T,$
\item $ \forall \ j=0,\ldots, n_{k+1} -1, \quad  Y'_{k,j} \rightarrow Y'_{k,j+1}a_{k+1, n_{k+1} - j},$
 \item $ Y'_{k,n_k} \rightarrow X_k,$

\end{itemize}
where   $Y_{i,j}$ and $Y'_{k,j}$  are new variables not in $V$.

 Now remove the production $p$ from $P$ and add respectively to $P$ and $V$ the productions defined above and
the corresponding set of new variables $Y_{i,j}$'s and $Y'_{k,j}$'s.

By  applying the previous argument to every production $p$ of the latter form, we will get a new grammar $G' = (V', T', I', P', S'),$
where $I'=I$, $S'=S$ and the sets $V'$ and $P'$ are obtained from $V$ and $P$ respectively by iterating the latter combinatorial transformation.

It is useful now to remark that, in correspondence of every production 
 $Xf \rightarrow u_1 X_1 \cdots u_kX_k u_{k+1}$,    of $G$ of the form (1) or (3), there exists
 a derivation of $G'$ such that
 $$Xf \Rightarrow_ {G'}^*  w_1 X_1 w_2  X_2\cdots w_k X_k w_{k+1},$$
 where, for all $i=1,\dots, k+1$, $w_i \in (T')^*$ and  $w_i \in \pi_{T} ^{-1}(u_i)$.   
  
 Taking into account the latter argument, the form of the new productions added to $G'$,
  and Eq. (\ref{eqproj}), by induction on the length of the derivations of $G$ and $G'$ respectively,  one proves 
 the following two claims:
  \begin{itemize}
  \item for every $w'\in T'^*$, $S' \Rightarrow_{G'} ^* w'$ if and only if there exists a derivation of $G$ 
  such that $S \Rightarrow_G ^* w$, with $w\in T^*$, and  $w'\in \pi_{T} ^{-1}(w)$.
\item if a non negative integer  bounds
the index of an arbitrary  derivation of $G$  the same does for $G'$. This implies that $G'$ is a finite-index grammar.
\end{itemize}
This concludes the proof. 
     \qed\end{proof}

Next, it is possible to show closure under rational transductions.
  \begin{lemma}\label{lm:2-trans}  
  Let $T$ and $T'$ be two alphabets. Let $\tau : T^* \rightarrow (T')^*$ be a rational transduction from $T^*$ into $(T')^*$. 
  If $L$ is a language of $T^*$ in the family $\FIN$, then $\tau (L) \in \FIN$.
  \end{lemma}
\begin{proof}
 Let us first assume that $T \cap T'=\emptyset$. By Nivat's theorem for the representation of rational 
transductions (see \cite{Be}, Ch. III, Thm 4.1), there exists a regular set $R$ over the
alphabet $(T \cup T')$ such that
$$
\tau = \{(\pi_T (u), \pi _ {T'}(u)): \ u \in R\},
$$
where $\pi_T$ and $\pi _ {T'}$ are the projections of $(T\cup T')^*$ onto $T^*$ and $T'^*$ respectively.
 
 From the latter, one has that, for every $u\in T^*$, $\tau (u) =  \pi _ {T'} (\pi _ T ^ {-1} (u) \cap R),$ so that
\begin{equation}\label{eq1:lm:2}  
\tau (L) = \bigcup _{u\in L} \tau (u) =   \pi _ {T'} (\pi _ T ^ {-1} (L) \cap R).
\end{equation}
Since, by hypothesis, $L \in \FIN$, the claim follows from (\ref{eq1:lm:2}), by applying Lemma \ref{lm:1},   
\ref{lm:3}, and \ref{lm:2}.
 
  Let us finally treat the case where $T$ and $T'$ are not disjoint. Let $T''$ be a copy of $T'$ with $T'' \cap T = \emptyset$
 and let $c_{T''}: (T')^* \rightarrow (T'')^*$ be the corresponding copying iso-morphism from $(T')^*$ onto $(T'')^*$. 
 Since $T'' \cap T = \emptyset$, by applying the latter argument to
 the rational transduction $ c_{T''}\tau : T^* \rightarrow (T'')^*$, one has $( c_{T''} \tau) (L) \in \FIN$.
 Since $c_{T''} ^{-1}(( c_{T''} \tau) (L) ) = \tau (L)$, then the claim follows from the latter
 by applying Lemma \ref{lm:1}.      \qed\end{proof}

Since  inverse morphisms are rational transductions,  the following is immediate:
 \begin{corollary}\label{lm-cor}  
$\FIN$ is closed under inverse morphisms. 
\end{corollary}
By Lemma \ref{lm:1}, Lemma \ref{lm:3}, and Corollary \ref{lm-cor},   we obtain:
\begin{theorem} The family $\FIN$ is a full trio.
\label{trios}
\end{theorem}


We now prove a result which extends the semi-linearity
of a family of languages to a bigger family. 
If $\cal C$ is a full trio of semi-linear languages and
$\cal L$ is the family of languages $\LL(\NCM)$ accepted by $\NCM$s,
let ${\cal C}  \wedge {\cal L} = \{L_1 \cap L_2  : L_1 \in {\cal C},
 L_2 \in {\cal L}\}$.
%
%
\begin{theorem}\label{prop-3}
Let ${\cal C}$ be a full trio of semi-linear languages. 
Every language in ${\cal C}  \wedge {\cal L}(\NCM)$ has a semi-linear Parikh image.
\end{theorem}

\begin{proof} Let $A$ and $B$ be disjoint alphabets.  Consider the homomorphism $$\widehat{\pi_{A}}: (A \cup B)^*  \rightarrow A^*$$
defined before Lemma \ref{lm:2}.
If $L$ is a language over $A^*$, then $\widehat{\pi_{A}}^{-1} (L) = \{ x :  x \in (A \cup B)^*, h(x) \in L\}$.

Let $A = \{a_1, \ldots, a_n\}$ and $L_1 \subseteq A^*$ be in $\cal C$. 
Then $\widehat{\pi_{A}}^{-1}(L_1)$ is also in $\cal C$, since $\cal C$ is closed under inverse homomorphism.
Note that the Parikh image of $L_1$, $\psi(L_1)$, is semi-linear since $\cal C$ is a semi-linear family.

Now let $L_2 \subseteq A^*$ be a language accepted by an $\NCM$ $M_2$. Any $\NCM$ can be simulated by an $\NCM$ $M_2$ whose counters are $1$-reversal \cite{BB1974}. We may assume that a string is accepted by $M_2$ if and only if it enters a unique halting state $f$  with all counters zero.

Let $M_2$ have $k$ $1$-reversal counters.
Let $B = \{p_1, q_1, \ldots, p_k,q_k\}$ be new symbols disjoint from $A$. Construct an NFA $M_3$ which when given a string $w$ in $(A \cup B)^*$ simulates $M_2$, but whenever counter $c_i$ increments, $M_3$ reads the next input symbol and checks that it is $p_i$.   When $M_2$ decrements counter $c_i$,  $M_3$  reads $q_i$ from the input.
(Note that after the first $q_i$ is read,  no $p_i$ should appear on the remaining input symbols.)
$M_3$ guesses when each counter $c_i$ becomes zero (this may be different time for each $i$),
after which, $M_3$ should no longer read $q_i$. At some point, $M_3$ guesses that all counters are zero.
It continues the simulation and when $M_2$ accepts in state $f$, $M_3$ accepts.
Clearly, a string $x$ in $A^*$ is accepted by $M_2$ if and only if there is a string
$w$ in $(A \cup B)^*$ accepted by $M_3$ such that:
\medskip

\noindent
(1) $\widehat{\pi_{A}} (w) = x,$\\
(2) $|w|_{p_i} = |w|_{q_i}$ for each  $1 \leq i \leq k$.
\medskip

\noindent
Let $R_3$ be the regular set accepted by $M_3$.
Since $\cal C$ is a full trio:
$$\widehat{\pi_{A}}^{-1} (L_1) \in {\cal C}, \quad L_4 = ( \widehat{\pi_{A}}^{-1} (L_1) \cap R_3) \in {\cal C}.$$
Hence the Parikh image of $L_4, \psi(L_4),$ is a semi-linear set $Q_4.$

Now $A = \{a_1, \ldots, a_n\}$ and $B = \{p_1, q_1, \ldots, p_k, q_k\}$.
Define the semi-linear set $$Q_5 = \{ (s_1, \ldots,  s_n, t_1, t_1, \ldots,  t_k, t_k ) : s_i,  t_i \geq 0\}.$$
(Note that the first $n$ coordinates refer to the counts 
corresponding to symbols $a_1, \ldots,  a_n,$  and the last $2k$ coordinates
refer to the counts corresponding to symbols $(p_1,q_1, \ldots, p_k, q_k)$.)

Then $Q_6 = Q_4 \cap Q_5$ is semi-linear,  since semi-linear sets are closed under intersection.
Now $\psi(L_1 \cap  L_2)$ coincides with the projection of $Q_6$ on the first $n$ coordinates.
Hence $\psi(L_1 \cap L_2)$ is semi-linear, since semi-linear sets are closed
under projections.
\qed
\end{proof}

Note that the above proposition does not depend on how the languages in
$\cal C$ are specified.
It extends the semi-linearity of languages
in $\cal C$ to a bigger family that can do some  ``counting''.
The theorem applies to all well-known full trios
of semi-linear languages, in particular, to ${\cal C} = \UFIN$.

\begin{corollary}
 Let $\cal C$ be  a full trio whose closures under homomorphism, inverse homomorphism and intersection with regular sets are effective.
Moreover, assume  that for each $L$ in ${\cal C}$, $\psi(L)$ can effectively be constructed.
Then ${\cal C}  \wedge {\cal L}(\NCM)$ has a decidable emptiness problem.
 \end{corollary}
 Indeed, decidability of emptiness follows immediately from effective construction of the semilinear set \cite{CIAA2016} as
having any vector describe a linear set implies the language is non-empty, and no vector implies the language is empty.
Note that $\cal L$ is also a full trio of semi-linear languages.  It is easy to see that the
theorem is not true if  $\cal L$ is an arbitrary full trio of semi-linear languages.
For example suppose ${\cal C} = {\cal L}$ is the family of languages accepted by 1-reversal
NPDAs (= linear context-free languages).  Let
$$L _1 = \{a^{n_1} \sharp \cdots \sharp a^{n_k} \$ a^{n_k} \sharp  \cdots \sharp a^{n_1}  : k \geq 4 , n_i \geq 1\},$$
$$L_2  = \{a^{n_1} \sharp \cdots \sharp  a^{n_k} \$ a^{m_k} \sharp \cdots \sharp a^{m_1}  : k \geq 4, \ n_i, \ m_i \geq 1, \ m_j= n_{j+1}, \ 1 \leq j < k\}.$$

Clearly, $L_1$ and $L_2$ can be accepted by 1-reversal NPDAs.
But $L_ 1 \cap L_2$ is  
$\{ (a^n\sharp)^{k-1} a^n \$ (a^n\sharp)^{k-1} a^n   :  n \geq 1, \ k \geq 4\}$ and it is not semi-linear.

 Similarly, it is known that the theorem does not hold
 when ${\cal C} = {\cal L}$ is the
 family of languages accepted by NFAs with one unrestricted counter (i.e., NPDAs
  with a unary stack alphabet in addition to a distinct bottom of the stack symbol which
  is never altered), as similar languages $L_1, L_2$ in this family can be constructed such that their intersection
  is not semilinear (Proposition 31 of \cite{EIMShuffle}).
  
Finally, let ${\cal C}_1$ and ${\cal C}_2$ be any full trios of semi-linear 
languages.  It is
clear that ${\cal C}_1  \cup {\cal C}_2 = \{L_1 \cup L_2 ~:~ L_1 \in {\cal C}_1, L_2 \in {\cal C}_2 \}$
is a semi-linear family.  One can also show that
${\cal C}_1 \cdot {\cal C}_2 = \{L_1 L_2 ~:~ L_1 \in {{\cal C}_1}, L_2 \in 
{{\cal C}_2} \}$ is
a semi-linear family.
%
%
%

\section{Bounded Languages and Hierarchy Results}

The purpose of this section is to demonstrate that all bounded Ginsburg semi-linear languages are in $\UFIN$ (thus implying they are in $\FIN$ as well), but not in $\ILIN$.

Notice that the language $L$ from the remarks following Theorem \ref{initialhierarchy} is a bounded Ginsburg semi-linear
language.
Thus, the following is true:
\begin{theorem}
\label{linearbounded}
There are bounded Ginsburg semi-linear languages that are not
in $\ILIN$.
\end{theorem}
Furthermore, it has been shown that in every semi-linear full trio,
all bounded languages in the family are bounded Ginsburg semi-linear
\cite{DLT2016}. Further, $\ILIN$ is a semi-linear full trio
\cite{DP}. Therefore, the bounded languages in $\ILIN$ are strictly
contained in the bounded languages contained in any family containing
all bounded Ginsburg semi-linear languages. We only mention here
three of the many such families mentioned in \cite{DLT2016}.
\begin{corollary}
The bounded languages in $\ILIN$ are strictly contained in the bounded
languages from $\LL(\NCM), \LL(\DCM), \ETOLfin$.
\end{corollary}

\begin{theorem}
\label{boundedGinsburg}
$\UFIN$ contains all bounded Ginsburg semi-linear languages.
\end{theorem}
\begin{proof}
  We now prove that if $L$ is a bounded Ginsburg semi-linear language,
  with $L \subseteq u_1^* \cdots u_k^*$, then $L\in \UFIN$. 
    By Definition \ref{+}, $L= \varphi (B)$, where $B$ is a semi-linear subset of $\N^k$. 
  Since $\UFIN$ is closed under union by Lemma \ref{lm-0}, it is enough
to show it for a linear set $B$.
 Let    $B$ be  a set of the form
$B= \{{\bf b}_0 + x_1{\bf b}_1 + \cdots + x_\ell{\bf b}_{\ell} \ : x_1, \ldots, x_\ell \in \N\}, $
  where ${\bf b}_0, {\bf b}_1, \ldots, {\bf b}_{\ell},$ are  vectors of $\N^k$. By denoting the arbitrary vector 
  ${\bf b}_i$ as  $(b_{i1}, \ldots, b_{ik})$, we write $B$ as
  $$\{({ b}_{01} + x_1{ b}_{11} + \cdots + x_\ell{ b}_{\ell1}, \ \ldots, \ { b}_{0k} + x_1{ b}_{1k} + \cdots + x_\ell{ b}_{\ell k}), \  : x_1, \ldots, x_\ell \in \N\}, $$
  so that    the language $L = \varphi (B)$ becomes
  
  \begin{equation}\label{eq-L-1}
u_1 ^{{ b}_{01} + x_1{ b}_{11} + \cdots + x_\ell{ b}_{\ell1}}u_2 ^{{ b}_{02} + x_1{ b}_{12} + \cdots + x_\ell{ b}_{\ell2}} \cdots u_k ^ { { b}_{0k} + x_1{ b}_{1k} + \cdots + x_\ell{ b}_{\ell k}},
\end{equation}
  where $x_1, \ldots, x_\ell \in \N$. 
  Let us now define an indexed grammar $G$  such that $L=L(G)$. Let $G=(V, T, I, P, S)$, where 
  $$ V = \{S, {Y}, X_1, \ldots, X_k\}, \quad T=A,  \quad I = \{e, f_1, f_2, \ldots, f_\ell\},$$ and the set $P$ of productions   is the following:
  \begin{enumerate}
\item $P_{start}=(S \rightarrow {Y}e)$
\item For every $j=1,\ldots, \ell,$ $P_j=({Y} \rightarrow {Y}f_j)$
\item $Q=({Y}  \rightarrow X_1 X_2 \cdots X_k)$
\item For every $i=1, \ldots, k$ and for every $j=1, \ldots, \ell,$  $$R_{i0}=(X_i e \rightarrow u_i ^ {b_{0i}}), \quad R_{ij}=(X_i f_j \rightarrow u_i ^ {b_{ji}}X_i).$$

\end{enumerate}
Let us finally prove that $L=L(G)$ and $G$ is an uncontrolled grammar. 
 Let us first show  that $L\subseteq L(G)$. Let $w\in L$. By (\ref{eq-L-1}), there exist $x_1, \ldots, x_\ell\in \N$ such that
$$w=u_1 ^{{ b}_{01} + x_1{ b}_{11} + \cdots + x_\ell{ b}_{\ell1}}u_2 ^{{ b}_{02} + x_1{ b}_{12} + \cdots + x_\ell{ b}_{\ell2}} \cdots u_k ^ { { b}_{0k} + x_1{ b}_{1k} + \cdots + x_\ell{ b}_{\ell k}}.$$
Consider the derivation defined by the word over the alphabet $P$:
 $${\cal P} = P_{start} P_1 ^ {x_1}P_2 ^ {x_2}\cdots P_\ell ^ {x_\ell}Q Q_1\cdots Q_k,$$
where, for every $i=1,\ldots, k$,
$Q_i =R_{i\ell}^{x_\ell} \cdots R_{i2}^{x_2} R_{i1}^{x_1} R_{i0}.$
It is easily checked that  $S \Rightarrow _{\cal P} w$. Indeed,
\\
$S \Rightarrow _{P_{start}}  {Y}e  \Rightarrow _{P_1 ^ {x_1}P_2 ^ {x_2}\cdots P_\ell ^ {x_\ell}} Y f_\ell ^ {x_\ell} \cdots  f_1 ^ {x_1} e   
  \Rightarrow _  Q $\\
 $ X_1  f_\ell ^ {x_\ell} \cdots  f_1 ^ {x_1} e  \cdots X_k  f_\ell ^ {x_\ell} \cdots  f_1 ^ {x_1} e   \Rightarrow _  {Q_1} u_1 ^{{ b}_{01} + x_1{ b}_{11} + \cdots + x_\ell{ b}_{\ell1}}X_2  \cdots X_k  f_\ell ^ {x_\ell} \cdots  f_1 ^ {x_1} e  $\\\
 $\Rightarrow _  {Q_2} u_1 ^{{ b}_{01} + x_1{ b}_{11} + \cdots + x_\ell{ b}_{\ell1}}u_2 ^{{ b}_{02} + x_1{ b}_{12} + \cdots + x_\ell{ b}_{\ell2}}X_3  \cdots X_k 
 f_\ell ^ {x_\ell} \cdots  f_1 ^ {x_1} e    \Rightarrow _  {Q_3\cdots Q_k} w,$
 \medskip
 
 \noindent
  so that $w\in L(G)$.
Similarly, it can be shown that $L(G) \subseteq L$. Thus $L= L(G)$. Moreover, taking into account the form of the
productions of $G$, it is easily checked that the index of every derivation of $G$ is not larger than $k$. \qed

\end{proof}

Since it is known that in any semi-linear full trio, all bounded languages in the family are bounded Ginsburg semi-linear, the bounded languages in
$\UFIN$ coincide with several other families, including a deterministic machine model \cite{DLT2016}.
\begin{corollary}
The bounded languages in $\UFIN$ coincide with the following families of languages: bounded Ginsburg semi-linear languages, bounded languages 
in $\LL(\NCM), \LL(\DCM),$ $ \ETOLfin$, the class of string languages of simple (i.e., linear and non deleting)  tree grammars (see \cite{K2014}) 
and several other families listed in \cite{DLT2016}.
\end{corollary}

Also, since $\ILIN$ does not contain all bounded Ginsburg semi-linear languages by Theorem \ref{linearbounded}, but $\UFIN$ does, the following is immediate:
\begin{corollary}
The bounded languages in $\ILIN$ are strictly contained in the bounded
languages of $\UFIN$.

\end{corollary}

Next, a restriction of $\UFIN$ is studied and compared to the other families.
 And indeed, this family is quite general as it contains all bounded
Ginsburg semi-linear languages in addition to some languages that
are not in $\ETOLfin$.
 
Now let $p=(Af\rightarrow \nu)\in P$, with $f\in I \cup \{\lambda\}$, be a production. Then $p$ is called {\em special} if the number
 of occurrences of variables of $V$ in $\nu$ is at least $2$, and {\em linear}, otherwise. Denote by $P_{\cal S}$ and $P_{\cal L}$ the sets
 of special and linear productions of $P$ respectively.  By Definition \ref{eq-def-4}, a
 grammar  $G$ is uncontrolled finite-index if and only if the number of times special productions appear
  in every successful derivation of $G$ is upper bounded by a given fixed integer (not depending on the derivation). 
 
Next, we will deal with  uncontrolled grammars such that in every successful derivation of $G$, at most one special production occurs.
  The languages generated by such grammars form a family  denoted $\UFINONE$. It is worth noticing that a careful  rereading of the proof of 
 Theorem \ref{trios} and Lemma \ref{lm-0} shows that they hold for
 $\UFINONE$ as well. Further, it is clear that only one special production
 is used in every derivation of a word in the proof of Theorem \ref{boundedGinsburg}. Therefore, the following holds:
   \begin{theorem}\label{trio-restrict} 
  The family $\UFINONE$  is a union-closed full trio and it contains all bounded Ginsburg semi-linear languages. 
\end{theorem}
It is immediate from the definitions that
$\ILIN \subseteq \UFINONE \subseteq \UFIN$. Further, since $\UFINONE$
contains all bounded Ginsburg semi-linear languages by Theorem
\ref{trio-restrict}, but the linear indexed languages do not, by
Theorem \ref{linearbounded}, the following holds:
\begin{theorem}
$\ILIN \subset \UFINONE \subseteq \UFIN$.
\end{theorem}

Then the following is true from \cite{DLT2016}.
 \begin{corollary}
$\UFINONE$ is a semi-linear full trio containing all bounded Ginsburg semi-linear languages. Further, the bounded languages in $\UFINONE, \ \UFIN, \ \LL(\NCM), \ \LL(\DCM),$ and $ \ETOLfin$ all coincide, (also with several others listed in \cite{DLT2016}).
\end{corollary}

\section{Some Examples, Separation, and Decidability Results}\label{examples}

We start this section by giving an example that clarifies previous results.
{
  \begin{example}
  Let $L=\{a^nb^nc^n\$a^nb^nc^n : n\in \N\}$. If $\varphi : \N^7 \rightarrow a^*b^*c^*\$^*a^*b^*c^*$, then
  $L= \varphi (B)$, where $B = \{{\bf b}_0 + n {\bf b}_1 : n\in \N\}$, with
  ${\bf b}_0 = (0,0,0,1,0,0,0)$ and  ${\bf b}_1 = (1,1,1,0,1,1,1)$. It is worth noticing that, by the discussion preceding Theorem 
  \ref{linearbounded}, $L$ is not a linear
  indexed language.     
  We define an uncontrolled finite-index grammar
  $G = (V, T, I, P, S)$ where $V= \{S, Y, X_1,X_2,X_3,X_4,$ $X_5,X_6,X_7 \}$, $T= \{a, b, c, \$\},$ $I= \{e, f\},$  and the set $P$ of productions is:
 $$
 P_{start}=S \rightarrow {Y}e, \,
  P={Y} \rightarrow {Y}f, \, Q={Y}  \rightarrow X_1 X_2 \cdots X_7$$
 \[
\begin{array}{lllllll}
   X_1f \rightarrow a X_1 & X_2f \rightarrow b X_2 & X_3f \rightarrow cX_3 & X_4f \rightarrow  X_4    
  & X_5f \rightarrow a X_5 & X_6f \rightarrow b X_6  \\
   X_7f \rightarrow c X_7  &  X_1e \rightarrow \lambda & X_2 e \rightarrow \lambda & X_3 e \rightarrow \lambda & X_4 e \rightarrow  \$   
  & X_5 e \rightarrow \lambda \\
  X_6 e \rightarrow \lambda & X_7 e \rightarrow \lambda . & & & & & 
\end{array}
\]
For an arbitrary derivation, we get 
$$S \, \Rightarrow\,  Ye \,  \Rightarrow ^n \, Y f ^n e 
=X_1 f ^n e X_2 f ^n e X_3 f ^n e X_4 f ^n e X_5 f ^n e X_6 f ^n e X_7 f ^n e\,    \Rightarrow ^* \, a^nb^nc^n\$a^nb^nc^n.$$
As the only freedom in derivations of $G$ consists  of how many times the rule $P$ is applied and of trivial variations in order to perform the rules
$X_i f  \rightarrow \sigma X_i, \sigma \in T \cup \{\varepsilon\}$, it should be clear that $L = L(G)$.
  \end{example} }

It is known that decidability of several properties holds for semi-linear trios where the properties are effective \cite{CIAA2016}.
This is the case for $\UFIN$, and also for $\ILIN$ \cite{DP}.
\begin{corollary}
Containment, equality, membership, and emptiness are decidable for bounded languages in $\UFIN$ and $\ILIN$.
\end{corollary}

Lastly, it is known that $\ETOLfin$ cannot generate some context-free languages \cite{Rozoy}, but all context-free languages can be generated by indexed linear grammars by Theorem \ref{thm-A}, which are all in $\UFINONE$.
\begin{corollary}
\label{notETOL}
There are languages in $\UFINONE$ and $\ILIN$ that are not in $\ETOLfin$.
\end{corollary}

  We provide an example of language in $\FIN$ whose Parikh image is not  a semi-linear set.

  \begin{example}\label{ex:1}
  We construct a grammar of index $3$, which is not uncontrolled,  that generates the language $L = \{aba^2 b \cdots a^n b a^{n+1} \ : \ n\geq 1\}$.
  Let $G = (V, T, I, P, S)$ be the grammar where $V = \{S, A, B, X, X', X''\}$, $T = \{a,b\}$, $I = \{e, f, g\}$, and the set of productions of $G$ are defined as:
  \begin{itemize}
\item $p_0=S \rightarrow Xe, \quad p_1=X \rightarrow ABX'f, \quad p_2= X' \rightarrow X, \quad p_3=X' \rightarrow X'',$
\item $p_4=X''f \rightarrow aX'', \quad p_5=X''e \rightarrow  a, \quad p_6=Af \rightarrow aA, \quad p_7=Ae \rightarrow\lambda,$
\item $p_8=Bf \rightarrow B, \quad p_9=Be \rightarrow b.$
 \end{itemize}
 One can check that  $G$ satisfies the properties mentioned above.
 
Let $G'$ be the grammar obtained from $G$ by replacing the production $p_9$ above with $(Be \rightarrow \lambda).$
Then one verifies that $G'$ is a grammar of index $3$ generating the unary language 
$\{a ^{n(n+1)/2} \ : \ n\geq 2\}$, that is not bounded Ginsburg semi-linear.
\end{example}

%
%

From Example \ref{ex:1} we get

\begin{corollary}
There are languages in $\FIN$ that are not semi-linear. Furthermore, there are
bounded (and unary) languages in $\FIN$ that are not bounded
Ginsburg semi-linear.
\label{notsemilinear}
\end{corollary}

This allows for the separation of $\UFIN$ (which only contains semi-linear languages) and $\FIN$.
\begin{corollary}
$\CFL \subset \ILIN \subset \UFINONE \subseteq \UFIN \subset \FIN$.
\end{corollary}

  Finally, we show that all finite-index $\ETOL$ languages are finite-index indexed languages.
 \begin{theorem}
$\ETOLfin \subset \FIN$.
\end{theorem}

\begin{proof}
Strictness follows since $\FIN$ contains non-semi-linear languages
by Corollary \ref{notsemilinear}, however $\ETOLfin$ only
contains semi-linear languages \cite{RozenbergFiniteIndexETOL}.

We refer to \cite{RozenbergFiniteIndexETOL} for the formal definitions of $\ETOL$ systems and finite-index $\ETOL$ systems, which we will omit.

Let $G= (V,{\cal P}, S, T)$ be a $k$-index $\ETOL$ system. 
We can assume without loss of generality that $G$ 
is in so-called active-normal form, so that the set of active 
symbols of $V$ (those that can be changed by 
some production table) is equal to $V \setminus T$. 
Let ${\cal P} = \{f_1, \ldots, f_r\}$ be the set of production tables.
Then create an indexed grammar $G' = (V', T, I, P, S')$ where
$V' = (V \setminus T) \cup \{S'\}$, $S'$ is a new variable,
$I = \{f_1, \ldots, f_r\}$, and $P$ contains the following 
productions:
\begin{enumerate}
\item $S' \rightarrow S' f_i, \forall i, 1 \leq i \leq r,$
\item $S' \rightarrow S$,
\item $Bf_i \rightarrow \nu, \forall (B \rightarrow \nu) \in f_i, B \in V \setminus T$.
\end{enumerate}

Let $w \in L(G)$. Then $w_0 \Rightarrow_{f_{j_1}} w_1 \Rightarrow \cdots \Rightarrow_{f_{j_l}} w_l, w_0 = S, w_l = w$.
Let $w_i'$ be obtained from $w_i$ by placing 
$f_{j_{i+1}} \cdots f_{j_l}$ after each variable of $w_i$.

We will show by induction on $i$, $0 \leq i \leq l$, that
$S' \Rightarrow_{G'}^* w_i'$. Indeed, 
$S' \Rightarrow_{G'}^* Sf_{j_1} \cdots f_{j_l} = w_0'$, by using 
productions of type 1 followed by 2. Assume the inductive
hypothesis is true for some $i$, $0 \leq i < l$. Then 
$S' \Rightarrow_{G'}^* w_i'$. Then the next index on every variable of 
$w_i'$ is $f_{j_{i+1}}$. Applying the productions corresponding to those used in
the derivation $w_i \Rightarrow_{f_{j_{i+1}}} w_{i+1}$
in table $f_{j_{i+1}}$ on each variable of $w_i'$ one at a time from 
left-to-right created in 3.\ of the construction above, $w_{i+1}'$ is obtained.
It is also clear that if the original derivation is of index-$k$,
then the resulting derivation is of index-$2k$
(since the derivation of the indexed grammar proceeds sequentially
instead of in parallel, the number of variables of the indexed grammar could potentially be more than $k$, but it is always
less than the number of variables in the sentential form of the $\ETOL$
system plus the next sentential form).

Let $w \in L(G')$. Thus, $w_0 \Rightarrow_{p_1} w_1 \Rightarrow_{p_2}
\cdots \Rightarrow_{p_l} w_l$, where $S' = w_0$ and $w_l = w \in T^*$.
It should also be clear that we can assume without loss of generality
that this derivation proceeds by rewriting variables in a
``sweeping left-to-right'' manner. That is, if $w_i = w_i' B w_i''$
derives $w_{i+1}$ by rewriting variable $B$, then $w_{i+1}$
derives $w_{i+2}$ by rewriting the first variable of $w_i''$ if
it exists, and if not, the first variable of $w_{i+1}$.
Then one ``sweep'' of the variables by rewriting each variable is
similar to one rewriting step of an $\ETOL$ system. This is akin to a breadth first
traversal on the derivation tree of $w$.

By the construction, there exists $\alpha > 0 $ such that
$p_1, \ldots, p_{\alpha}$ are productions created in step 1,
$p_{\alpha+1}$ is created in step 2, and $p_{\alpha+2}, \ldots, p_l$
are created in step 3.
Let $\beta_1, \ldots, \beta_q$ be such that $\beta_1 = \alpha+2$, and the
derivation from $w_{\beta_i}$ is the start of the $i$th ``sweep''
from left-to-right,
and let $\beta_{q+1} = l$.
For $1 \leq i \leq q+1$, let $u_i$ be obtained from $w_{\beta_i}$
by removing all indices (so $u_{q+1} = w_l$).

We will show by induction that for all $i$, $1 \leq i \leq q+1$,
it is true that $S \Rightarrow_G^* u_i$, and all variables in $w_{\beta_i}$
are followed by the same index sequence. Indeed,
$w_{\beta_1} = w_{\alpha+2} = S\gamma$ for some $\gamma \in I^*$,
$u_1 = S$, and $S \Rightarrow_G^* u_1 = S$.
Assume that the inductive hypothesis holds for some $i$, 
$1 \leq i \leq q$. Then in $w_{\beta_i}$, all variables are
followed by the same index sequence. Let $f$ be the first index following
every variable. Then in the subderivation 
$w_{\beta_i} \Rightarrow_{p_{\beta_i}} \cdots \Rightarrow_{p_{\beta_{i+1}}} w_{\beta_{i+1}}$, because all productions applied were created in
step 3, they must all pop the first index, and since they all start
with the same index, they must all have been created from productions in the same table
$f$. It is clear that $u_i \Rightarrow_G u_{i+1}$ using production
table $f$. It is also immediate
that all variables in $w_{\beta_{i+1}}$ are followed by the
same sequence of indices. The proof follows.
\qed\end{proof}

It is an open question though as to how $\ETOLfin$
compares to $\UFIN$. For finite-index $\ETOL$, uncontrolled systems,
defined similarly to our definition of uncontrolled, is identical to finite-index $\ETOL$. Furthermore, it is known that $\ETOLfin$ is closed under Kleene-$*$ \cite{RozenbergFiniteIndexETOL}
and therefore contains $\{a^nb^n c^n : n>0\}^*$. But
we conjecture that this language is not in $\UFIN$ despite
being in $\FIN$ by the proposition above. This would imply that
$\UFIN$ is incomparable with $\ETOLfin$ by Corollary \ref{notETOL}.
Finally we observe that it would be interesting to know whether the inclusion
$ \UFINONE\subseteq \UFIN $ is strict or not. The examples presented in this paper would suggest that the two latter families could be equal.

We finally note that the class of linear indexed languages studied by Duske and Parchmann is a proper subset of the one studied by Gazdar and Vijayy-Shanker. 
An example of a language in the second class but not in the first, is the language $L = \{a^n b^n c^n \$ a^m b^m c^m \mid n,m \geq 0\}$
which appears in the remarks following Theorem \ref{First-hierarchy}. It might then also be interesting to study the class of Gazdar and Vijayy-Shanker in connection with finite index restrictions.

\bibliographystyle{splncs03}

\begin{thebibliography}{}
\expandafter\ifx\csname url\endcsname\relax
  \def\url#1{\texttt{#1}}\fi
\expandafter\ifx\csname urlprefix\endcsname\relax\def\urlprefix{URL }\fi
\expandafter\ifx\csname href\endcsname\relax
  \def\href#1#2{#2} \def\path#1{#1}\fi

\end{thebibliography}


\begin{thebibliography}{99}
 
 


 \bibitem{A} A. V. Aho, Indexed grammars---an extension of context-free grammars,
	 J. ACM,  	15 (4), 647--671, (1968).

 \bibitem{A1} A. V.  Aho, Nested stack automata, J. ACM,  	16, 383--406, (1969).

 \bibitem{BB1974} B. S. Baker, R. V. Book, Reversal-Bounded Multipushdown Machines,  
 J. Comput. Syst. Sci., 8, 315--332, (1974).



     \bibitem{Be}  J. Berstel,  {\em Transductions and Context-Free Languages},
   B.B. Teubner, Stuttgart, 1979.
   
    \bibitem{DP1989} J. Dassow,   Gh. P\v{a}un,
{\em Regulated Rewriting in Formal Language Theory,} 
EATCS Monographs on Theoretical Computer Science, 18, Springer-Verlag, Berlin, 1989.

 \bibitem{DP}    
  J. Duske, R. Parchmann, Linear indexed grammars, Theoret. Comput. Sci. 32, 
47--60, (1984). 

 \bibitem{EIMShuffle}
J. Eremondi, O. H. Ibarra, and I. McQuillan, On the Complexity and Decidability of Some
Problems Involving Shuffle, Information and Computation, 259 (2),  214--224, (2018).


 
 
\bibitem{Gazdar1988} G. Gazdar,  Applicability of Indexed Grammars to Natural Languages,
pp. 69--94,   Springer Netherlands, Dordrecht, (1988).

 
   
 \bibitem{Gins} S. Ginsburg, {\em The Mathematical Theory of
Context-free Languages}, Mc Graw- Hill, New York, 1966.



\bibitem{C3.6} M. A. Harrison,  {\em Introduction to Formal Language Theory},  
Addison-Wesley Publishing Co., Reading, Mass., 1978. 



\bibitem{C4.8} 
J. E. Hopcroft, J. D. Ullman, {\em  Introduction to  Automata Theory, 
Languages and Computation}, Addison-Wesley Publishing Co., Reading, Mass., 1979.
 
 
 

 \bibitem{Ibarra1978}
O. H. Ibarra, Reversal-bounded multicounter machines and their decision problems, J. ACM, 1 (25),   116--133,
 (1978).





\bibitem{CIAA2016}  O.H. Ibarra and I. McQuillan,
 On bounded semilinear languages, counter machines, and finite-index {ET0L}. In: Y. Han and K. Salomaa (eds.),
 Lecture Notes in Computer Science, 21st International Conference on Implementation and Application of Automata, CIAA 2016, Seoul, South Korea,
	vol. 9705, pp. 138--149, (2016).

 


 
\bibitem{DLT2016}  O.H. Ibarra and I. McQuillan,
 On families of full trios containing counter machine languages. In: S. Brlek and C. Reutenauer (eds.),
 Lecture Notes in Computer Science, 20th International Conference on Developments in Language Theory, DLT 2016, Montreal, Canada,
	vol. 9840, pp. 216--228, (2016).


\bibitem{K2014}  M. Kanazawa,
 A Generalization of Linear Indexed Grammars Equivalent to Simple Context-Free Tree
 Grammars. In: G. Morrill, R. Muskens, R. Osswald, F. Richter (eds.),
 Lecture Notes in Computer Science, 19th International Conference, FG 2014, T\"ubingen, Germany, 2014,
	vol. 8612, pp. 86--103, (2014).



\bibitem{RS} G. Rozenberg and A. Salomaa, 
	    {\em The Mathematical Theory of L Systems}, Academic Press, Inc., 
	 New York, 1980.



\bibitem{RozenbergFiniteIndexETOL}
 G. Rozenberg and D. Vermeir, On {ET0L} systems of finite index,
Information and Control, 38, 103--133, (1978).


\bibitem{RozenbergFiniteIndexGrammars}
 G. Rozenberg and D. Vermeir, On the effect of the finite index restriction on several families of grammars,
 Information and Control, 39, 284--302, (1978).



 \bibitem{Rozoy} B. Rozoy, The {Dyck} language {$D_1^{\prime *}$} is not generated by any matrix grammar of finite index, 
 Information and Computation 74 (1),  64--89, (1987).




 \bibitem{Vijay-Shanker1994} K. Vijay-Shanker and D. J. Weir, 
The equivalence of four extensions of context-free grammars,
Mathematical Systems Theory,  27 (6), 511--546, (1994).

\bibitem{ICALP2015}  G. Zetzsche,
An Approach to Computing Downward Closures. In: M.M. Halld{\'o}rsson and K. Iwama and N. Kobayashi and B. Speckmann (eds.),
 Automata, Languages, and Programming: 42nd International Colloquium, ICALP 2015, Kyoto, Japan, July 6-10, 2015, Proceedings, Part II,
	pp. 440--451, (2015).


 
\end{thebibliography}

\end{document}